\documentclass[11pt]{article}
\usepackage[letterpaper,margin=1in]{geometry}
\usepackage{setspace}
\usepackage[utf8]{inputenc}

\usepackage{subcaption}
\usepackage{physics}
\usepackage{amsmath,graphicx}
\usepackage{mathtools}
\usepackage{algorithmic}
\mathtoolsset{showonlyrefs=true}
\usepackage{float}
\usepackage{amsmath}
\usepackage{amsfonts}
\usepackage{amsthm}
\usepackage{enumitem}  
\usepackage{xfrac}
\usepackage[ruled,vlined]{algorithm2e}
\usepackage{authblk}

\usepackage[backref=page]{hyperref}
\newif\ifbackrefshowonlyfirst
\backrefshowonlyfirstfalse
\makeatletter
\let\BR@direct@old@hyper@natlinkstart\hyper@natlinkstart
\renewcommand*{\hyper@natlinkstart}{\phantomsection\BR@direct@old@hyper@natlinkstart}
\let\BR@direct@oldBR@citex\BR@citex
\renewcommand*{\BR@citex}{\phantomsection\BR@direct@oldBR@citex}%

\long\def\hyper@page@BR@direct@ref#1#2#3{p. \hyperlink{#3}{#1}}

\ifx\backrefxxx\hyper@page@backref
    \let\backrefxxx\hyper@page@BR@direct@ref
    \ifbackrefshowonlyfirst
    \fi
\else
    \ifbackrefshowonlyfirst
    \fi
\fi

\RequirePackage{etoolbox}
\patchcmd{\Hy@backout}{Doc-Start}{\@currentHref}{}{\errmessage{I can't seem to patch backref}}
\makeatother

\usepackage[numbers]{natbib}
\usepackage{blindtext} 


\newtheorem{theorem}{Theorem}[section]

\newtheorem{lemma}[theorem]{Lemma}
\newtheorem{corollary}[theorem]{Corollary}

\newtheorem{remark}[theorem]{Remark}

\newcommand{\inner}[2]{ \left\langle #1 \; , \; #2 \right\rangle }
\usepackage{interval}
\intervalconfig{soft open fences}
\usepackage{xcolor}
\usepackage{braket}
\usepackage{changes}

\author[1]{Chun-Neng~Chu}
\author[1,2]{Wei-Fu~Tseng}
\author[1,2,3]{Yen-Huan Li}

\affil[1]{Department of Computer Science and Information Engineering,\protect\\National Taiwan University}
\affil[2]{Department of Mathematics, National Taiwan University}
\affil[3]{Center for Quantum Science and Engineering, \protect\\ National Taiwan University}

\usepackage{etoolbox}
\newbool{versionArXiv}
\setbool{versionArXiv}{true} 

\begin{document}

\ifbool{versionArXiv}{
\date{}
\title{Algorithms for Computing the Petz-Augustin Capacity}
}{
\title{Algorithms for Computing \\ the Petz-Augustin Capacity}
}

\ifbool{versionArXiv}{}{\author{%
      \IEEEauthorblockN{
        Chun-Neng Chu\textsuperscript{*}, 
        Wei-Fu Tseng\textsuperscript{*\S}, 
        and Yen-Huan Li\textsuperscript{*\S\P}
      }
      \IEEEauthorblockA{
        \textsuperscript{*}Department of Computer Science and Information Engineering, National Taiwan University\\
        \textsuperscript{\S}Department of Mathematics, National Taiwan University\\
        \textsuperscript{\P}Center for Quantum Science and Engineering, National Taiwan University
      }
      }}

\maketitle
\begin{abstract}
  \ifbool{versionArXiv}{
  }{
      THIS PAPER IS ELIGIBLE FOR THE STUDENT PAPER AWARD.

  }
  We propose the first algorithms with non-asymptotic convergence guarantees for computing the Petz-Augustin capacity, which generalizes the channel capacity and characterizes the optimal error exponent in classical-quantum channel coding. 
  This capacity can be equivalently expressed as the maximization of two generalizations of mutual information: the Petz-R\'{e}nyi information and the Petz-Augustin information.
  To maximize the Petz-R\'{e}nyi information, we show that it corresponds to a convex H\"{o}lder-smooth optimization problem, and hence the universal fast gradient method of Nesterov (2015), along with its convergence guarantees, readily applies. 
  Regarding the maximization of the Petz-Augustin information, we adopt a two-layered approach: we show that the objective function is smooth relative to the negative Shannon entropy and can be efficiently optimized by entropic mirror descent;
  each iteration of entropic mirror descent requires computing the Petz-Augustin information, for which we propose a novel fixed-point algorithm and establish its contractivity with respect to the Thompson metric. 
  Notably, this two-layered approach can be viewed as a generalization of the mirror-descent interpretation of the Blahut-Arimoto algorithm due to He et al. (2024).  

  \ifbool{versionArXiv}{
  }{

      A full version of this paper is available at https://arxiv.org/abs/XXXXX.
  }
\end{abstract}

\section{Introduction}
\label{sec:introduction}
Shannon's seminal work \cite{Shannon1948} introduced the notion of channel capacity, establishing that the optimal transmission error probability vanishes in the asymptotic limit of infinite blocklengths, provided that the code rate remains below the channel capacity.
The channel capacity is formulated as the maximization of the mutual information over input distributions,
and can be efficiently computed by the well-known Blahut-Arimoto algorithm \cite{Blahut1972,Arimoto1972}.

Since blocklengths are finite in practice, it is essential to characterize the optimal transmission error probability in finite-length regimes. 
This leads to the notion of the Augustin capacity, which generalizes the channel capacity and characterizes the optimal error exponent in classical channel coding \cite{Gallager1965,Augustin1978,Csiszar1995,Nakiboglu2019}. 
Subsequently, several algorithms have been proposed to compute the Augustin capacity \cite{Arimoto1976,Jitsumatsu2020,Kamatsuka2024b,Kamatsuka2024c,Kamatsuka2025}.

The concepts of channel capacity and Augustin capacity, as well as their operational meanings, have subsequently been extended to classical-quantum (CQ) channels; these are known, respectively, as the quantum channel capacity \cite{Schumacher1997,Holevo1998} and the Petz-Augustin capacity \cite{Dalai2017,Cheng2019,Li2025b,Renes2025,Cheng2025}. 
However, while the Blahut-Arimoto algorithm has been generalized to compute the quantum channel capacity \cite{Nagaoka1998,Li2019b,Ramakrishnan2021,Hayashi2024,He2024a,Hayashi2025}, the computation of the Petz-Augustin capacity remains unexplored.

To address this gap, we study the computation of the Petz-Augustin capacity of order $\alpha$, denoted by $C_{\mathcal{W},\alpha}$, for general CQ channels $\mathcal{W}$.
This capacity can be equivalently expressed as the maximization of two generalizations of the quantum mutual information, known, respectively, as the Petz-R\'{e}nyi information of order $\alpha$, denoted by $I_{\alpha}^{\text{R}}(p,\mathcal{W})$, and the Petz-Augustin information of order $\alpha$, denoted by $I_{\alpha}^{\text{A}}(p,\mathcal{W})$ \cite{Cheng2022}.
We propose two approaches tailored to each maximization problem and show that their convergence rates exhibit different advantages depending on $\alpha$.
Our approaches apply to the regime $\alpha\in[1/2,1)$, where $C_{\mathcal{W},\alpha}$ characterizes the random-coding exponent for CQ channel coding \cite{Li2025b,Renes2025,Cheng2025}.
The computation for the range $\alpha\in(0,1/2)$, which is relevant to the sphere-packing exponent \cite{Dalai2017,Cheng2019}, is left for future work.

\subsection{Problem Formulation}
For a general classical-quantum (CQ) channel $\mathcal{W}$ mapping from an input alphabet $\Set{1,\dots,n}$ to the set $\mathcal{D}_d$ of quantum states in $\mathbb{C}^{d \times d}$, the Petz-Augustin capacity $C_{\mathcal{W},\alpha}$ of order $\alpha\in(0,1)$ is defined via the following optimization problems over the probability simplex $\Delta_n$ in $\mathbb{R}^n$ \cite[Equation~(22)]{Cheng2022}: 
\begin{align}
    C_{\mathcal{W},\alpha} 
    &\coloneqq -\min_{p\in\Delta_n}g_{\alpha}^{\text{R}}(p)\label{def:RenyiCapacity}\\
    &=-\min_{p\in\Delta_n} g_{\alpha}^{\text{A}}(p)\label{def:AugCapacity},
\end{align}
where $g_{\alpha}^{\text{R}}(p) \coloneqq -I_{\alpha}^{\text{R}}(p,\mathcal{W})$ and $g_{\alpha}^{\text{A}}(p) \coloneqq -I_{\alpha}^{\text{A}}(p,\mathcal{W})$ are the negative Petz-R\'{e}nyi and Petz-Augustin information, respectively.
The Petz-R\'{e}nyi information and the Petz-Augustin information are defined as:
\begin{align}
    I_{\alpha}^{\text{R}}(p,\mathcal{W})
    \coloneqq \frac{\alpha}{\alpha-1} \log \left(\mathrm{Tr}\left[\left(\sum_{j=1}^n p[j] \mathcal{W}(j)^{\alpha}\right)^{1/\alpha}\right]\right),\nonumber
\end{align}
and
\begin{align}
    \label{def:AugInfo}
    I_{\alpha}^{\text{A}}(p,\mathcal{W})
    \coloneqq  \min_{Q\in\mathcal{D}_d}\sum_{j=1}^n p[j] D_{\alpha}\left(\mathcal{W}(j)\|Q\right),
\end{align}
where $p[j]$ denotes the $j^{\text{th}}$ coordinate of $p$.
Here, the Petz-R\'{e}nyi divergence $D_{\alpha}(A\|Q)$ is given by \cite{Petz1986}:
\begin{align}
    D_{\alpha}(A \| Q) 
    \coloneqq  \frac{1}{\alpha - 1} \log \left(\mathrm{Tr}\left[A^{\alpha} Q^{1-\alpha}\right]\right),\nonumber
\end{align}
whenever it is well-defined, 
and $D_\alpha ( A \| Q ) \coloneqq  \infty$ otherwise. 

\subsection{Main Results}

\subsubsection{H\"{o}lder-Smooth Optimization}
In Section~\ref{sec:MaximizingRenyiInfo}, we propose a convex H\"{o}lder-smooth optimization approach for solving the optimization problem \eqref{def:RenyiCapacity}.

Since $g_{\alpha}^{\text{R}}(p)$ in \eqref{def:RenyiCapacity} is only known to be quasi-convex \cite[Proposition~(b)]{Cheng2022},
standard convex optimization methods are not directly applicable.
To address this, we exponentiate the original objective function $g_{\alpha}^{\text{R}}(p)$, leading to the convex optimization problem:
\begin{align}
    \label{def:ExpRenyiCapacity}
    \min_{p\in\Delta_n}\tilde{g}_{\alpha}^{\text{R}}(p),
    \quad \tilde{g}_{\alpha}^{\text{R}}(p)
    \coloneqq\mathrm{Tr}\left[\left(\sum_{j=1}^n p[j] \mathcal{W}(j)^{\alpha}\right)^{1/\alpha}\right],
\end{align}
where the convexity of $\tilde{g}_{\alpha}^{\text{R}}(p)$ follows from the fact that the $r^{\text{th}}$ power of the Schatten $r$-norm,  $X\mapsto \norm{X}_r^r$ is convex for $r\geq 1$.

\ifbool{versionArXiv}{
    Furthermore, we show in Appendix~\ref{appendix:SVM} that the classical case of \eqref{def:ExpRenyiCapacity} corresponds to a generalized support vector machine (SVM) problem studied by Rodomanov et al. \cite{Rodomanov2024}.
}{
    Furthermore, we show in the full version \cite{Chu2026} that the classical case of \eqref{def:ExpRenyiCapacity} corresponds to a generalized support vector machine (SVM) problem studied by Rodomanov et al. \cite{Rodomanov2024}. }
They pointed out that this generalized SVM problem satisfies the H\"{o}lder smoothness condition \cite{Nesterov2018a}, a property generalizing the Lipschitz continuity and Lipschitz gradient conditions  in the convex optimization literature, which allows H\"{o}lder-smooth optimization methods to efficiently solve this problem.
By leveraging matrix analysis tools, we prove that this H\"{o}lder smoothness condition still holds for the general quantum setting in \eqref{def:ExpRenyiCapacity}.
Consequently, existing H\"{o}lder-smooth optimization methods \cite{Rodomanov2024,Nemirovskii1985,Devolder2014,Nesterov2015,Li2025,Zhao2025} and their convergence guarantees become applicable.
By combining these methods with our analysis of the
H\"{o}lder smoothness parameters of $\tilde{g}_{\alpha}^{\text{R}}(p)$, we obtain a convergence rate of $O\left(T^{1-1.5/\alpha}\right)$
for computing
$C_{\mathcal{W},\alpha}$ when $\alpha\in[1/2,1)$, where $T$ denotes the number of iterations.

Curiously, while the Petz-Augustin capacity generalizes the channel capacity, which can be computed by the Blahut-Arimoto algorithm, the algorithm yielded by the H\"{o}lder-smooth approach differs significantly from the Blahut-Arimoto algorithm. 
Furthermore, its convergence rate approaches $O(T^{-0.5})$ as $\alpha$ approaches $1$, which is significantly slower than the well known $O(T^{-1})$ rate of the Blahut-Arimoto algorithm \cite{Arimoto1972}. 
This motivates us to propose the following algorithm, which, along with the convergence analysis, was originally developed in our unpublished preprint \cite[Section~6]{Chu2025}.

\subsubsection{Blahut-Arimoto-Type Algorithm}
In Section~\ref{sec:MaximizingAugInfo}, we propose a Blahut-Arimoto-type algorithm for solving the optimization problem \eqref{def:AugCapacity}, namely, the minimization of the negative Petz-Augustin information.
Unlike the Petz-R\'{e}nyi information, the Petz-Augustin information $I_{\alpha}^{\text{A}}(p,\mathcal{W})$ defined in \eqref{def:AugInfo} does not admit a closed-form expression. 
Consequently, we decompose the optimization problem \eqref{def:AugCapacity} into two nested subproblems.

For the inner subproblem \eqref{def:AugInfo}, we propose a novel fixed-point algorithm for computing $I_{\alpha}^{\text{A}}(p,\mathcal{W})$ for $\alpha\in(1/2,1)\cup(1,\infty)$.
By establishing its contractivity with respect to the Thompson metric \cite{Thompson1963}, 
we prove that it converges at a linear rate of $O(\abs{1-1/\alpha}^{T})$.
Building on this result, we can compute $\varepsilon$-approximate values of $g_{\alpha}^{\text{A}}(p)$ and its gradient $\nabla g_{\alpha}^{\text{A}}(p)$ in $O(\log(1/\varepsilon))$ time.

Then, we treat the evaluations of $g_{\alpha}^{\text{A}}(p)$ and $\nabla g_{\alpha}^{\text{A}}(p)$ as black-box oracles and apply entropic mirror descent (EMD) \cite{Ben-Tal2001,Beck2003} to solve the outer subproblem \eqref{def:AugCapacity} for orders $\alpha\in(1/2,1)$.
Our approach is inspired by the recent work of He et al. \cite{He2024a}, who interpret the Blahut-Arimoto algorithm as EMD.
They further establish a relative smoothness property \cite{Bauschke2017,Lu2018}---a generalization of the Lipschitz gradient condition---for the negative mutual information.
This property directly yields the same convergence rate of $O(T^{-1})$ for the Blahut-Arimoto algorithm as that originally established by Arimoto \cite{Arimoto1972}.
In this work, we prove that this relative smoothness property
also holds for the negative Petz-Augustin information $g_{\alpha}^{\text{A}}(p)$.
This ensures that EMD is applicable, and its convergence rate of $O(T^{-1})$ complements the $O(T^{1-1.5/\alpha})$ rate of our H\"{o}lder-smooth optimization approach when $\alpha$ is close to $1$.

Notably, our Blahut-Arimoto-type algorithm has an additional advantage in the study of channel coding with input constraints---for instance, the optimal error exponents of constant composition codes for CQ channels can be
characterized by the Petz-Augustin information with fixed input distributions \cite{Dalai2017,Cheng2025,Shi2025}.  
He et al. \cite{He2024a} have proposed a primal-dual hybrid gradient (PDHG) method to compute the channel capacity under linear input constraints.
This method adapts the Blahut-Arimoto algorithm by incorporating additional dual updates to handle the linear constraints, and its convergence analysis leverages the relative smoothness property of the negative mutual information.
Since the same relative smoothness property for the negative Petz-Augustin information is established in this work, this PDHG method can be readily applied to compute the Petz-Augustin capacity \eqref{def:AugCapacity} under the same class of input constraints.
In contrast, it is unclear how to apply our H\"{o}lder-smooth approach to this constrained setting.

\section{Related Work}
\label{sec:RelatedWork}
This work is the first to propose algorithms for computing the Petz-Augustin capacity with non-asymptotic convergence guarantees.
Although several algorithms have been proposed for computing the classical Augustin capacity \cite{Arimoto1976,Jitsumatsu2020,Kamatsuka2024b,Kamatsuka2024c,Kamatsuka2025}, their extension to the quantum setting remains unclear.

Notably, a classical channel can be viewed as a special case of the CQ channel where all output states are diagonal matrices. 
Thus, our Blahut-Arimoto-type algorithm is directly applicable to the classical setting. 
It is the first to establish a non-asymptotic convergence guarantee for computing the Augustin capacity via maximizing the Augustin information, whereas the algorithm proposed by Kamatsuka et al. \cite{Kamatsuka2024b} only provides an asymptotic guarantee.

For the convex H\"{o}lder-smooth optimization, existing works have applied this framework to problems such as linear regression, logistic regression, and SVMs \cite{Rodomanov2024,Nesterov2015,Li2025}.
However, these problems do not involve quantum states in their formulations.
To address the challenges arising from the non-commutativity of quantum states, we leverage matrix analysis tools to prove that the objective function in the optimization problem \eqref{def:ExpRenyiCapacity} satisfies the H\"{o}lder smoothness condition.
This enables us to apply existing H\"{o}lder-smooth optimization methods to compute the Petz-Augustin capacity.

Finally, for computing the Petz-Augustin information, our method is the first to provide a non-asymptotic convergence guarantee.
While entropic mirror descent with Armijo line search \cite{Li2019a} or with the Polyak step size \cite{You2022} is applicable to this problem, it only guarantees asymptotic convergence.
Although several algorithms with non-asymptotic convergence guarantees have been proposed for computing the classical Augustin information \cite{Augustin1978,Tsai2024,Wang2024,Kamatsuka2025,Karakos2008}, it remains unclear how to extend these methods to the quantum setting.
For a more comprehensive discussion, we refer readers to our unpublished preprint \cite[Section~2]{Chu2025}, which provides a detailed comparison between our algorithm and existing approaches for computing the Augustin information.

\section{Preliminaries}
\label{sec:Preliminaries}
This section introduces the notation used throughout this work and the convex optimization frameworks employed in the design and analysis of our algorithms.

\subsection{Notations}
Let $\mathcal{B}_d$ denote the set of Hermitian matrices in $\mathbb{C}^{d \times d}$, and let $\mathcal{B}_{d,+}$ and $\mathcal{B}_{d,++}$ denote the corresponding positive semidefinite and positive definite cones, respectively.
For any function $f:\mathbb{R}\to\mathbb{R}$, we define its vector and matrix extensions as follows: for $x\in\mathbb{R}^n$, $f(x)\coloneqq \left(f(x[1]),\dots,f(x[n])\right)$; for $X\in\mathcal{B}_d$ with eigendecomposition $X=\sum_{i=1}^d \lambda_i u_i u_i^{*}$, $f(X)\coloneqq \sum_{i=1}^d f(\lambda_i) u_i u_i^{*}$.
Let $\norm{\cdot}_p$ denote the $\ell_p$-norm for vectors and the Schatten $p$-norm for matrices.
For any norm $\norm{\cdot}$ on $\mathbb{R}^n$, $\norm{\cdot}_{*}$ denotes its dual norm.
For $X\in\mathbb{C}^{d\times d}$,  $\abs{X}$ denote the modulus of $X$.
Let $\mathbf{1}_n$ denote the all-ones vector in $\mathbb{R}^n$.
Let $\odot$ denote the coordinate-wise product.
For any set $\mathcal{C}$, let $\mathrm{relint}(\mathcal{C})$ denote the relative interior of $\mathcal{C}$.
The Bregman divergence generated by $h:\mathbb{R}^n\to\mathbb{R}$ is defined as $B_h(y, x)\coloneqq h(y)-h(x)-\inner{\nabla h(x)}{y-x}$ for $x,y\in\mathbb{R}^n$.

\subsection{Mirror Descent}
Mirror descent is a standard first-order framework for convex optimization.
Specifically, consider minimizing a convex and differentiable function $f(x)$ over a closed convex set $\mathcal{C}\subseteq\mathbb{R}^n$.
This method computes the next iterate from the current iterate $z$ toward a direction $-v$ with step size $\eta>0$ via the following mapping:
\begin{equation}
  \label{alg:BregProx}
  T_{h,\mathcal{C}}(z, v, \eta)
  \in \arg\min_{x\in\mathcal{C}}\inner{v}{x} + \frac{1}{\eta}B_h(x,z).
\end{equation}
Notably, when $h(x)=\norm{x}_2^2/2$ and $v=\nabla f(z)$, the iteration rule \eqref{alg:BregProx} recovers the standard projected gradient method.

In this work, to handle the simplex constraint, we adopt the negative Shannon entropy $h(x)=\inner{x}{\log(x)}$, which is a widely used choice for such geometries.
This choice leads to the  following closed-form update:
\begin{equation}
  \label{alg:EMD}
  T_{h,\Delta_n}(z, v, \eta)
  =\frac{z\odot\exp\left(-\eta v\right)}{\norm{z\odot\exp\left(-\eta v\right)}_1}.
\end{equation}
By setting $v=\nabla f(z)$, the update rule \eqref{alg:EMD} recovers the well-known entropic mirror descent (EMD) algorithm \cite{Ben-Tal2001,Beck2003}.

\subsection{Relative Smoothness}
The relative smoothness condition \cite{Bauschke2017,Lu2018} is a generalization of the Lipschitz gradient condition in the convex optimization literature.
Specifically, a function $f$ is said to be $L$-smooth relative to a convex function $h$ on a convex set $\mathcal{C}\subseteq\mathbb{R}^n$  for some $L>0$ if $Lh-f$ is convex on $\mathcal{C}$.

Under the relative smoothness condition, mirror descent \eqref{alg:BregProx} enjoys the following convergence guarantee:
\begin{lemma}(\cite[Theorem~3.1]{Lu2018})
    \label{lemma:MirrorConv}
    Consider the problem of minimizing a convex function $f$ over the set $\mathcal{C}$,
    where \( f \) is $L$-smooth relative to a convex function \( h \) on \( \mathcal{C} \).
    Suppose that an iterative algorithm computes the next iterate by $x_{t+1} = T_{h,\mathcal{C}}(x_t, \nabla f(x_t), 1/L)$, with $x_1\in\mathrm{relint}(\mathcal{C})$.
    Then, for any $T\in\mathbb{N}$, we have
    \begin{equation}
        f(x_{T+1}) - f(x)
        \leq \frac{L\cdot B_{h}(x \| x_1)}{T},\quad\forall x\in\mathcal{C}.\nonumber
    \end{equation}
\end{lemma}

\subsection{H\"{o}lder Smoothness}
Another extension of the Lipschitz gradient condition is the H\"{o}lder smoothness condition \cite{Nesterov2018a}, which interpolates between the Lipschitz continuity and Lipschitz gradient conditions.
Specifically, a differentiable function $f$ is said to be  $(\nu,L_{\nu})$-H\"{o}lder smooth on a convex set $\mathcal{C}\subseteq\mathbb{R}^n$ with respect to a norm $\norm{\cdot}$ for $\nu\in[0,1]$ and $L_{\nu}>0$ if the following inequality holds:
\begin{align}
  \label{def:HolderSmooth}
  \norm{\nabla f(x)-\nabla f(y)}_{*} \leq L_{\nu}\norm{x-y}^{\nu},\quad\forall x,y \in \mathcal{C}.
\end{align}
Notably, the cases $\nu=0$ and $\nu=1$ in \eqref{def:HolderSmooth} recover the Lipschitz continuity and Lipschitz gradient conditions, respectively.

Subsequently, several H\"{o}lder-smooth optimization methods have been proposed \cite{Rodomanov2024,Nemirovskii1985,Devolder2014,Nesterov2015,Li2025,Zhao2025}, whose iteration complexities $O((L_{\nu}/\varepsilon)^{2/(1+3\nu)})$ for computing $\varepsilon$-approximate solutions to general convex H\"{o}lder-smooth optimization problems are known to be optimal \cite{Nemirovsky1983}.

\section{Maximizing the Petz-R\'{e}nyi Information}
\label{sec:MaximizingRenyiInfo}
This section presents our convex H\"{o}lder-smooth optimization approach for computing the Petz-Augustin capacity $C_{\mathcal{W},\alpha}$ via the maximization of the Petz-R\'{e}nyi information \eqref{def:RenyiCapacity}.
First, Lemma~\ref{lemma:ExpRenyiHolder} characterizes the H\"{o}lder smoothness of the objective function in the optimization problem \eqref{def:ExpRenyiCapacity}.
Second, Lemma~\ref{lemma:FGMConv} leverages this analysis and the universal fast gradient method (FGM) \cite{Nesterov2015} to establish convergence guarantees for solving \eqref{def:ExpRenyiCapacity}.
Finally, Corollary~\ref{cor:RenyiCapacityConv} translates these guarantees back to the maximization problem \eqref{def:RenyiCapacity}.

The following H\"{o}lder smoothness property is our key technical breakthrough in this section.
\begin{lemma}
  \label{lemma:ExpRenyiHolder}
  For any $\alpha\in[1/2,1)$, the function $\tilde{g}_{\alpha}^{\text{R}}$ in \eqref{def:ExpRenyiCapacity} is $((1-\alpha)/\alpha,1/\alpha)$-H\"{o}lder smooth on $\Delta_n$ with respect to $\norm{\cdot}_1$.
\end{lemma}

The proof of Lemma~\ref{lemma:ExpRenyiHolder} involves bounding the coordinate-wise differences of $\nabla \tilde{g}_{\alpha}^{\text{R}}$ at two arbitrary points $p_1,p_2\in\Delta_n$.
Our derivation relies on matrix generalizations of both the H\"{o}lder inequality \cite[Theorem~6.20]{Hiai2014} and the scalar inequality $\abs{x^r-y^r}\leq \abs{x-y}^r$ for $r\in\left(0,1\right]$ \cite[Theorem~6.35, Exercise~4.45]{Hiai2014}.
\ifbool{versionArXiv}{
    The full proof is provided in Appendix~\ref{appendix:ProofOfExpRenyiHolder}.
}{
Details are deferred to the full version \cite{Chu2026}.
}

Then, we apply FGM \cite{Nesterov2015}---an optimal convex H\"{o}lder-smooth optimization method---to solve \eqref{def:ExpRenyiCapacity}, and the following lemma states the resulting convergence guarantee.
\ifbool{versionArXiv}{
  We defer implementation details of FGM to Appendix~\ref{appendix:FGM}.
}{
  We leave implementation details of FGM to the full version \cite{Chu2026}.
}
\begin{lemma}[{\cite[Theorem~3]{Nesterov2015}}]
  \label{lemma:FGMConv}
  Consider applying FGM to solve the optimization problem \eqref{def:ExpRenyiCapacity} with $\alpha\in[1/2,1)$.
  Let the first iterate be $p_1=\mathbf{1}_n/n$, and fix an arbitrary error tolerance $\varepsilon>0$.
  Then, the output $p_{T+1}\in\Delta_n$ after $T$ iterations satisfies
  \begin{align} 
    \tilde{g}_{\alpha}^{\text{R}}(p_{T+1}) -\min_{p\in\Delta_n}\tilde{g}_{\alpha}^{\text{R}}(p)\leq \min_{\nu,L_{\nu}}\left(\frac{2^{2+4\nu}L_{\nu}^{2}}{\varepsilon^{1-\nu}T^{1+3\nu}}\right)^{1/(1+\nu)}\log(n)+\frac{\varepsilon}{2},\nonumber
  \end{align}
  where the minimization is taken over all parameters $(\nu,L_{\nu})$ for which $\tilde{g}_{\alpha}^{\text{R}}$ satisfies the H\"{o}lder smoothness condition \eqref{def:HolderSmooth} on $\Delta_n$ with respect to the $\ell_1$-norm.
\end{lemma}

\begin{remark}
  While several alternative methods achieve a similar convergence guarantee to that in Lemma~\ref{lemma:FGMConv} \cite{Rodomanov2024,Li2025,Zhao2025}, their analyses are restricted to the H\"{o}lder smoothness condition with respect to the $\ell_2$-norm. 
  A direct extension of Lemma~\ref{lemma:ExpRenyiHolder} to the $\ell_2$ setting yields the H\"{o}lder parameters $(1/\alpha-1,n^{0.5/\alpha}/\alpha)$.
  However, the additional dependence on $n$ would degrade the overall convergence rate, thereby justifying our focus on the $\ell_1$-norm analysis.
\end{remark}

By combining Lemma~\ref{lemma:FGMConv} with the H\"{o}lder smoothness parameters established in Lemma~\ref{lemma:ExpRenyiHolder} and setting the error tolerance to
\begin{align}
  \label{eq:FGM_epsilon}
  \varepsilon=\log(n)^{0.5/\alpha}T^{1-1.5/\alpha},
\end{align}
we obtain a convergence rate of $O\left(\log(n)^{0.5/\alpha}T^{1-1.5/\alpha}\right)$.

Finally, in Corollary~\ref{cor:RenyiCapacityConv}, we convert this guarantee for the auxiliary problem \eqref{def:ExpRenyiCapacity} into one for the original formulation \eqref{def:RenyiCapacity}, which directly follows from the definition $g_{\alpha}^{\text{R}}(p)=(\alpha/(1-\alpha))\log(\tilde{g}_{\alpha}^{\text{R}}(p))$ and the
fact that $\log(y)-\log(x)\leq (y-x)/x$ for any $y\geq x>0$.
\begin{corollary}
  \label{cor:RenyiCapacityConv}
  Let $p_{T+1}\in\Delta_n$ be as in Lemma~\ref{lemma:FGMConv} with $\varepsilon$ given in \eqref{eq:FGM_epsilon}.
  Then, for any $\alpha\in[1/2,1)$, we have
  \begin{align}
    {g}_{\alpha}^{\text{R}}(p_{T+1}) -\min_{p\in\Delta_n}{g}_{\alpha}^{\text{R}}(p)
    \leq C\cdot T^{1-1.5/\alpha},\nonumber
  \end{align}
  where $C=O(\log(n)^{0.5/\alpha}\exp((1/\alpha-1)C_{\mathcal{W},\alpha})/(1-\alpha))$.
\end{corollary}

\begin{remark}
  For a fixed number of iterations $T$, selecting the theoretically optimal $\varepsilon$ in Lemma~\ref{lemma:FGMConv} requires knowledge of the tightest H\"{o}lder smoothness parameters.
  The choice of $\varepsilon$ in \eqref{eq:FGM_epsilon} is based on the estimates provided in Lemma~\ref{lemma:ExpRenyiHolder} and may therefore be conservative; indeed, the numerical results in Section~\ref{sec:NumericalResults} suggest that employing a smaller $\varepsilon$ leads to better empirical performance.
\end{remark}

\section{Maximizing the Petz-Augustin Information}
\label{sec:MaximizingAugInfo}
This section presents our Blahut-Arimoto-type algorithm for computing the Petz-Augustin capacity $C_{\mathcal{W},\alpha}$ via the maximization of the Petz-Augustin information \eqref{def:AugCapacity}.
This algorithm was originally developed in our unpublished preprint \cite[Section~6]{Chu2025}.
Specifically, we decompose the optimization problem \eqref{def:AugCapacity} into two nested subproblems: an inner subproblem to compute the Petz-Augustin information \eqref{def:AugInfo} for a fixed input distribution $p$,
and an outer subproblem to maximize this quantity over the simplex $\Delta_n$.

\subsection{Inner Subproblem: Computing the Petz-Augustin Information}
An efficient algorithm for computing the Petz-Augustin information $I_{\alpha}^{\text{A}}(p,\mathcal{W})$ in the optimization problem \eqref{def:AugInfo} is crucial for implementing first-order methods to solve the outer subproblem \eqref{def:AugCapacity}.
Specifically, let $Q_{\star}(p)$ denote the minimizer of \eqref{def:AugInfo}.
The gradient $\nabla I_{\alpha}^{\text{A}}(p,\mathcal{W})$ with respect to $p$ can be expressed as \cite[Equation~(25)]{Cheng2022}:
\begin{align}
    \nabla I_{\alpha}^{\text{A}}(p,\mathcal{W})[j]
    = D_{\alpha}(\mathcal{W}(j)\| Q_{\star}(p)),\quad \forall j.\nonumber
\end{align}

To this end, we propose the following novel fixed-point algorithm for computing $I_{\alpha}^{\text{A}}(p,\mathcal{W})$ and its gradient for a fixed input distribution $p$ \cite[Equation~(11)]{Chu2025}:
\begin{itemize}
  \item Initialize $Q_1\in\mathrm{relint}\left(\mathcal{D}_d\right)$.
  \item For $t=1,2,\dots$, compute
  \begin{align}
    \label{alg:AugInfoIter}
    Q_{t+1}
    =\left(\sum_{j=1}^n p[j] \frac{\mathcal{W}(j)^{\alpha}}{\mathrm{Tr}\left[\mathcal{W}(j)^{\alpha} Q_t^{1-\alpha}\right]}\right)^{1/\alpha},
  \end{align}
  together with the estimates $\hat{\nabla}_{t+1}[k]=D_{\alpha}(\mathcal{W}(k)\| Q_{t+1})$ for each $k^{\text{th}}$ coordinate of $\nabla I_{\alpha}^{\text{A}}(p,\mathcal{W})$, and $\hat{I}_{t+1}=\inner{p}{\hat{\nabla}_{t+1}}$ for $I_{\alpha}^{\text{A}}(p,\mathcal{W})$.
\end{itemize}

To analyze the convergence rate of the algorithm \eqref{alg:AugInfoIter}, we adopt the Thompson metric $d_{\mathrm{T}}$ \cite{Thompson1963}, defined as
\begin{align}
        d_{\mathrm{T}}(V,U)\coloneqq \inf
        \{ r\geq0 \mid \exp(-r) V\leq U \leq \exp(r)V \},\quad\forall V,U\in\mathcal{B}_{d,++}.\nonumber
\end{align}
The following lemma establishes the contractive property of the algorithm \eqref{alg:AugInfoIter} with respect to the Thompson metric.
\begin{lemma}[{\cite[Lemma~5.2, Lemma~5.3]{Chu2025}}]
    \label{lemma:AugInfoContraction}
    For any $\alpha\in(1/2,1)\cup(1,\infty)$ and $p\in\mathrm{relint}\left(\Delta_n\right)$, the iteration rule \eqref{alg:AugInfoIter} satisfies
    \begin{align}
        d_{\mathrm{T}}(Q_{t+1}^{1-\alpha},Q_{\star}^{1-\alpha})
        \leq \abs{\frac{1-\alpha}{\alpha}} d_{\mathrm{T}}(Q_t^{1-\alpha},Q_{\star}(p)^{1-\alpha}).\nonumber
    \end{align}
\end{lemma}

Based on Lemma~\ref{lemma:AugInfoContraction}, we obtain the following convergence rate for approximating the first-order oracle of $I_{\alpha}^{\text{A}}(p,\mathcal{W})$.
\begin{lemma}[{\cite[Lemma~6.2]{Chu2025}}]
    \label{lemma:AugInfoConv}
    For any $\alpha\in(1/2,1)\cup(1,\infty)$ and $p\in\mathrm{relint}\left(\Delta_n\right)$, let $Q_{T+1},\hat{\nabla}_{T+1},\hat{I}_{T+1}$ be the outputs of the algorithm \eqref{alg:AugInfoIter} after $T$ iterations.
    Then, the errors, $\abs{\hat{I}_{T+1}-I_{\alpha}^{\text{A}}(p,\mathcal{W})}$ and $\norm{\hat{\nabla}_{T+1}-\nabla I_{\alpha}^{\text{A}}(p,\mathcal{W})}_{\infty}$, are both bounded by $O(\abs{1-1/\alpha}^{T}d_{\mathrm{T}}(Q_{1}^{1-\alpha},Q_{\star}(p)^{1-\alpha}))$.
\end{lemma}

\subsection{Outer Subproblem: Maximizing the Petz-Augustin Information}
We employ the following entropic mirror descent (EMD) iteration to solve the outer subproblem \eqref{def:AugCapacity} \cite[Equation~(14)]{Chu2025}. 
As discussed in Section~\ref{sec:introduction}, this approach can be viewed as a generalization of the Blahut-Arimoto algorithm.
\begin{itemize}
  \item Initialize $p_1=\mathbf{1}_n/n$ and choose the prox function $h(p)=\inner{p}{\log(p)}$.
  \item For $t=1,2,\dots$, compute
  \begin{align}
    \label{alg:AugCapacityIter}
    p_{t+1}
    =T_{h,\Delta_n}\left(p_t, \nabla g_{\alpha}^{\text{A}}(p_t), 1\right),
  \end{align}
  along with the function value $g_{\alpha}^{\text{A}}(p_{t+1})$.
\end{itemize}

For implementation, the update rule in \eqref{alg:AugCapacityIter} admits the closed-form expression in \eqref{alg:EMD}.
Moreover, according to Lemma~\ref{lemma:AugInfoConv} and the identity $g_{\alpha}^{\text{A}}(p) = -I_{\alpha}^{\text{A}}(p,\mathcal{W})$, both the function value and the gradient required in \eqref{alg:AugCapacityIter} can be efficiently computed using the inner algorithm \eqref{alg:AugInfoIter}.

To establish the convergence rate, we show that  $g_{\alpha}^{\text{A}}(p)$ satisfies the relative smoothness condition.
\begin{lemma}[{\cite[Lemma~6.3]{Chu2025}}]
    \label{lemma:AugInfoRelSmooth}
    For any $\alpha\in[1/2,1)$, the function $g_{\alpha}^{\text{A}}$ in \eqref{def:AugCapacity} is $1$-smooth relative to $h(p)=\inner{p}{\log(p)}$ on $\Delta_n$.
\end{lemma}
By combining Lemma~\ref{lemma:AugInfoRelSmooth} with the standard convergence guarantee for mirror descent (Lemma~\ref{lemma:MirrorConv}), we obtain the following result:
\begin{theorem}[{\cite[Theorem~6.5]{Chu2025}}]
  \label{thm:AugCapacityConv}
  Let $p_{T+1},g_{\alpha}^{\text{A}}(p_{T+1})$ be the outputs of \eqref{alg:AugCapacityIter} with $\alpha\in(1/2,1)$.
  Then, the optimization error satisfies 
  \begin{align}
    g_{\alpha}^{\text{A}}(p_{T+1}) - \min_{p\in\Delta_n}g_{\alpha}^{\text{A}}(p)
    \leq\frac{\log(n)}{T}.\nonumber
  \end{align}
\end{theorem}

\section{Numerical Results}
\label{sec:NumericalResults}
We evaluate our proposed methods over $T=1000$ iterations to compute the Petz-Augustin capacity $C_{\mathcal{W},\alpha}$ for $\alpha$ in $\{0.6,0.9\}$. 
We set the input alphabet size to $n=2^{7}$ and the dimension of the output quantum states to $d=2^{5}$.
The quantum states $\mathcal{W}(j)$ are generated using the \texttt{rand\_dm} function from the Python package QuTiP \cite{Johansson2012}.
All experiments are conducted on a machine with an
Intel Xeon Gold 5218 CPU running at 2.30 GHz.
\ifbool{versionArXiv}{
  The source code is available on GitHub\footnote{\url{https://github.com/chunnengchu/PetzAugustinCapacity/}}.
}{
  The source code is available on GitHub\footnote{https://github.com/chunnengchu/PetzAugustinCapacity/}.
}

In Figures~\ref{fig1} and \ref{fig2}, \emph{FGM--Balanced} denotes the FGM with $\varepsilon=\log(n)^{0.5/\alpha}T^{1-1.5/\alpha}$ following \eqref{eq:FGM_epsilon}, and  \emph{FGM--1e-9} denotes the FGM with a fixed $\varepsilon=10^{-9}$.
\emph{Blahut-Arimoto} refers to the proposed Blahut-Arimoto-type algorithm \eqref{alg:AugCapacityIter}.
Since the exact optimizer is unknown, we use the maximum value of $I_{\alpha}^{\text{R}}(p_t,\mathcal{W})$ observed over $3000$ iterations of FGM--1e-9 to approximate $C_{\mathcal{W},\alpha}$
when computing optimization errors.
\ifbool{versionArXiv}{
\begin{figure}[htbp]
    \centering

    \includegraphics[width=0.98\textwidth]{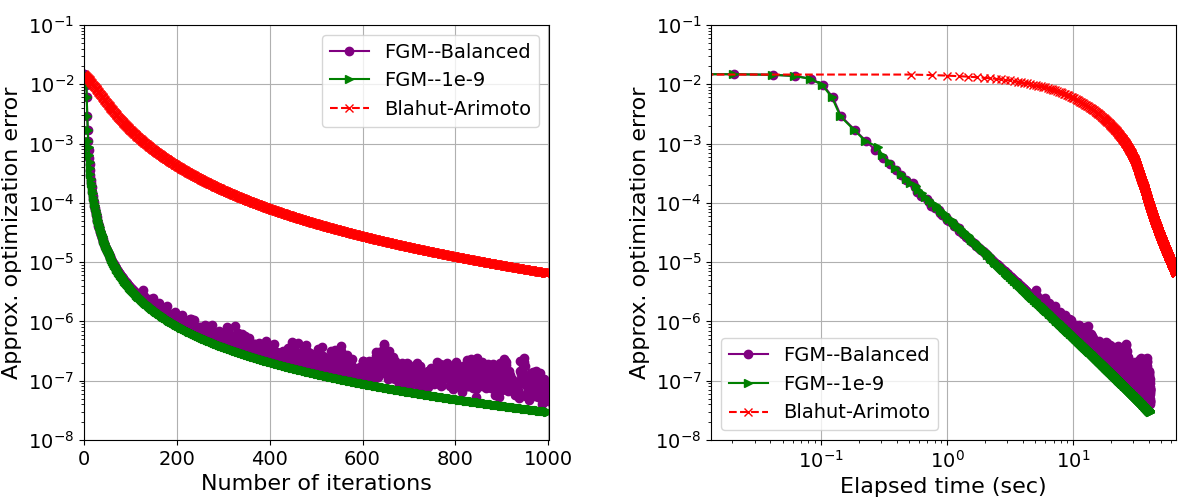}

    \caption{Computing Petz-Augustin capacity for order $\alpha = 0.6$.}
    \label{fig1}
\end{figure}
\begin{figure}[htbp]
    \centering
    
    \includegraphics[width=0.98\textwidth]{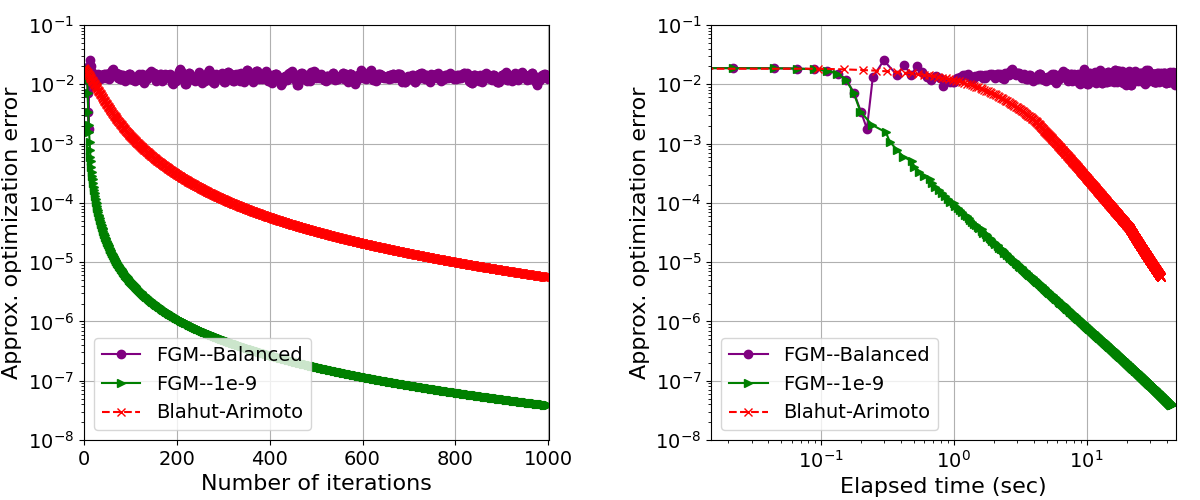}
    
    \caption{Computing Petz-Augustin capacity for order $\alpha = 0.9$.}
    \label{fig2}
\end{figure}
}{
\begin{figure}[htbp]
    \centering

    \includegraphics[width=0.49\textwidth]{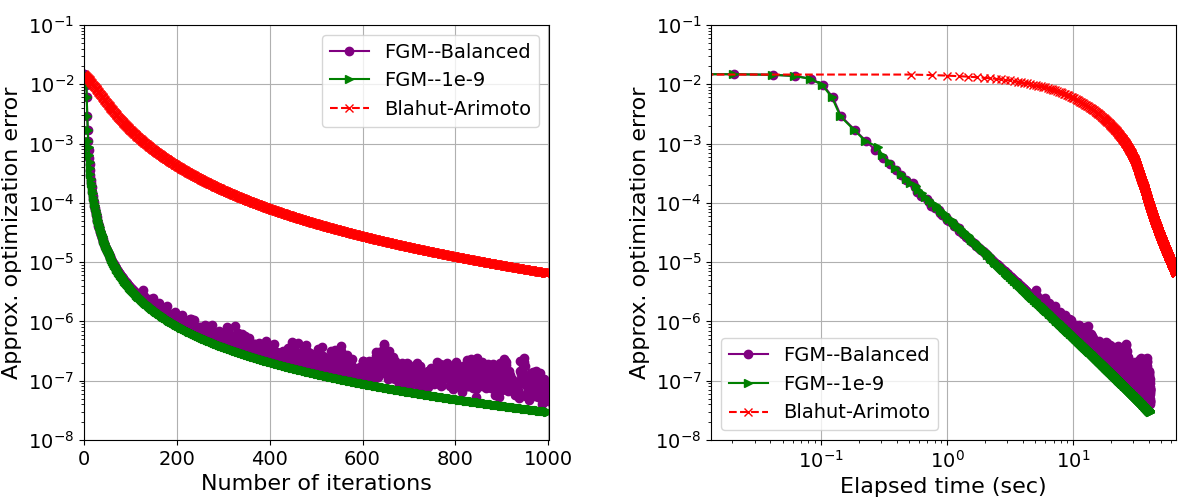}

    \caption{Computing Petz-Augustin capacity for order $\alpha = 0.6$.}
    \label{fig1}
\end{figure}
\begin{figure}[htbp]
    \centering
    
    \includegraphics[width=0.49\textwidth]{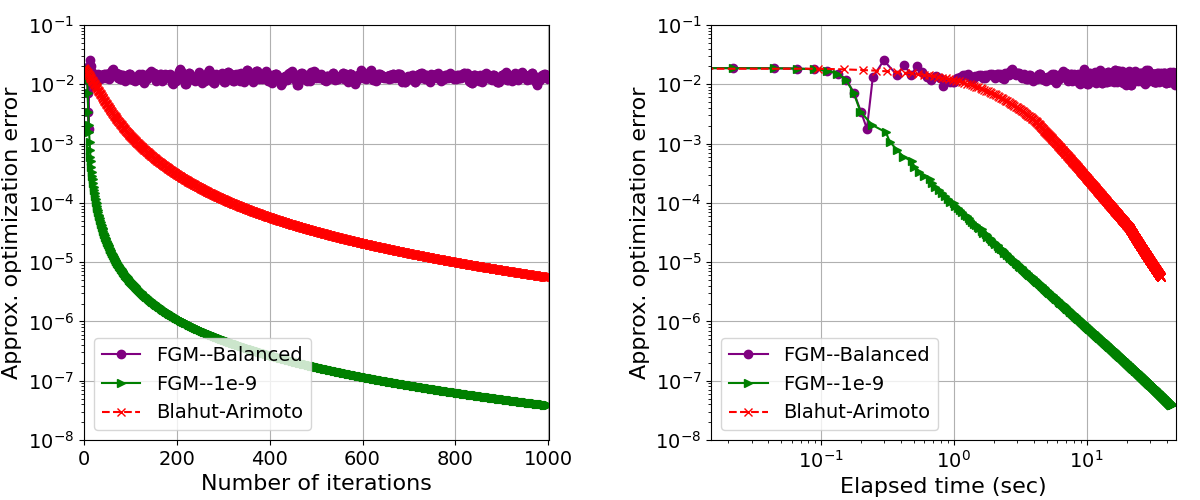}
    
    \caption{Computing Petz-Augustin capacity for order $\alpha = 0.9$.}
    \label{fig2}
\end{figure}

}

Figure~\ref{fig1} shows that FGM--Balanced outperforms the Blahut-Arimoto-type algorithm for $\alpha=0.6$.
Conversely, as shown in Figure~\ref{fig2}, the Blahut-Arimoto-type algorithm is superior for $\alpha=0.9$.
These observations are consistent with our theoretical guarantees:
for $\alpha=0.6$, the convergence rate $O(T^{-1.5})$ of FGM--Balanced  (Corollary~\ref{cor:RenyiCapacityConv}) is faster than the $O(T^{-1})$ rate of the Blahut-Arimoto-type algorithm (Theorem~\ref{thm:AugCapacityConv}); 
for $\alpha=0.9$, the rate $O(T^{-2/3})$ of FGM--Balanced
is slower, explaining its relative performance.

Notably, FGM--1e-9 has the best performance across both cases.
The choice $\varepsilon=10^{-9}$ can be interpreted as setting $\varepsilon=T^{-3}$.
Combining Lemma~\ref{lemma:FGMConv} with our H\"{o}lder smoothness estimates in Lemma~\ref{lemma:ExpRenyiHolder} yields an error bound of $O(T^{8\alpha-6})$ for FGM--1e-9.
This bound is clearly conservative for $\alpha=0.9$ relative to the empirical results, suggesting that both the H\"{o}lder smoothness estimates and the parameter selection in \eqref{eq:FGM_epsilon} can be further improved.

\section{Remarks}

After completing this work, we learned of a concurrent work by Lai and Cheng \cite{Lai2026}. 
Their work focuses on solving the auxiliary problem~\eqref{def:ExpRenyiCapacity}.
They show that the objective function is smooth relative to the negative Shannon entropy, which implies that EMD converges at a rate of $O(1/T)$ for all $\alpha \in ( 0, 1 )$.
	However, for $\alpha \in ( 1/2, 1 )$, the smoothness parameter they obtain depends on the inverse of the smallest eigenvalue of $\sum_{j=1}^n p[j] \mathcal{W}(j)^{\alpha}$ and can therefore be arbitrarily large; 
as a result, the corresponding optimization error bound involves a scalar that can be arbitrarily large (see Lemma~\ref{lemma:MirrorConv}).
Moreover, their approach does not address the constrained formulation.
By contrast, our approaches apply to the more restricted range $\alpha \in [ 1/2, 1 )$, but yield optimization error bounds that depend only on the problem parameters $\alpha$, $n$, and $C_{\mathcal{W},\alpha}$.
In addition, our second approach allows us to handle the constrained formulation.

\ifbool{versionArXiv}{
\bibliography{Contents/refs}
}{
\bibliographystyle{IEEEtran}
\bibliography{Contents/refs}
}

\ifbool{versionArXiv}{
}{
\enlargethispage{-1.2cm} 
}


\ifbool{versionArXiv}{
  \appendix 

\section{Connection with the Support Vector Machine}
\label{appendix:SVM}
This section introduces the connection between the classical case of the optimization problem \eqref{def:ExpRenyiCapacity} and a generalized support vector machine (SVM) problem considered by Rodomanov et al. \cite{Rodomanov2024}.

In particular, a general classical channel $\mathcal{W}$ can be viewed as a special case of classical-quantum channels in which all output states are diagonal matrices.
In this setting, the optimization problem \eqref{def:ExpRenyiCapacity} reduces to the following formulation:
\begin{align}
    \label{def:ClassicalExpRenyiCapacity}
    \min_{p\in\Delta_n} \tilde{g}_{\alpha}^{\text{R}}(p),
    \quad \tilde{g}_{\alpha}^{\text{R}}(p)
    \coloneqq\sum_{i=1}^d \left(\sum_{j=1}^n p[j] \tilde{A}_j[i]\right)^{1/\alpha},
\end{align}
where each $\tilde{A}_j$ denotes the vector of diagonal entries of the powered output state $\mathcal{W}(j)^{\alpha}$.

On the other hand, the generalized SVM problem studied by Rodomanov et al. \cite[Equation~(5)]{Rodomanov2024} is formulated as
\begin{align}
    \label{def:SVM}
    \min_{\norm{p}_2\leq 1} f_{\alpha}(p),
    \quad f_{\alpha}(p)
    \coloneqq \sum_{i=1}^d\max\Set{ \left(\left(\sum_{j=1}^n p[j] A_j[i]\right)-b[i]\right),0}^{1/\alpha},
\end{align}
where $A_j$ and $b$ are vectors in $\mathbb{R}^d$.
Here, $\alpha=1$ and $\alpha=0.5$ correspond, respectively, to the standard SVM with hinge loss and squared hinge loss.

By setting $b=0$ and $\tilde{A}_j = A_j$ for all $j$, the optimization problem \eqref{def:ClassicalExpRenyiCapacity} coincides with \eqref{def:SVM}, with a slight difference in the constraint set.
Notably, Rodomanov et al. \cite{Rodomanov2024} pointed out that $f_{\alpha}$ is $(\nu,L_{\nu})$-H\"{o}lder smooth with $\nu=(1-\alpha)/\alpha$ for $\alpha\in[1/2,1]$, and with the coefficient $L_{\nu}<\infty$ depending on $A_j$ and $b$.
This directly implies the H\"{o}lder smoothness of the classical case of $\tilde{g}_{\alpha}^{\text{R}}$ in \eqref{def:ClassicalExpRenyiCapacity}.
However, extending this analysis to the general quantum setting in \eqref{def:ExpRenyiCapacity} is non-trivial due to the non-commutative nature of quantum states.
We leverage matrix analysis tools to overcome this challenge, as detailed in Appendix~\ref{appendix:ProofOfExpRenyiHolder}.

\section{Proof of Lemma~\ref{lemma:ExpRenyiHolder}}
\label{appendix:ProofOfExpRenyiHolder}
We will use Lemma \ref{lemma:SchattenHolder} and Lemma \ref{lemma:MatrixPowerDiff} to prove Lemma~\ref{lemma:ExpRenyiHolder}.
\begin{lemma}[H\"{o}lder Inequality {\cite[Theorem~6.20]{Hiai2014}}]
  \label{lemma:SchattenHolder}
  For any $U,V\in\mathcal{B}_{d}$ and $1\leq p,q\leq \infty$ such that $1/p+1/q=1$, we have
    \begin{align}
        \norm{UV}_1
        \leq \norm{U}_p\norm{V}_q.
    \end{align}
\end{lemma}

\begin{lemma}[{\cite[Theorem~6.35, Exercise~4.45]{Hiai2014}}]
  \label{lemma:MatrixPowerDiff}
  For any $A,B\in\mathcal{B}_{d}^{+}$, $p>1$, and $r\in\left(0,1\right]$, we have
  \begin{align}
      \norm{A^{r}-B^{r}}_p
      \leq \norm{\abs{A-B}^{r}}_p.
  \end{align}
\end{lemma}

\begin{proof}[Proof of Lemma~\ref{lemma:ExpRenyiHolder}]
  Since the dual norm of the $\ell_1$-norm is the $\ell_{\infty}$-norm, it suffices to upper bound the coordinate-wise differences of $\nabla \tilde{g}_{\alpha}^{\text{R}}$ at any two points $p_1,p_2\in\Delta_n$.
  In particular, for any $j$, we write 
  \begin{align}
    &\abs{\nabla \tilde{g}_{\alpha}^{\text{R}}(p_1)[j]-\nabla \tilde{g}_{\alpha}^{\text{R}}(p_2)[j]}\nonumber\\
    &\quad=\frac{1}{\alpha} \abs{\mathrm{Tr}\left[\left(\left(\sum_{k=1}^n p_1[k]\mathcal{W}(k)^{\alpha}\right)^{1/\alpha-1} - \left(\sum_{k=1}^n p_2[k]\mathcal{W}(k)^{\alpha}\right)^{1/\alpha-1}\right)\mathcal{W}(j)^{\alpha}\right]}\nonumber\\
    &\quad\leq\frac{1}{\alpha}\norm{\left(\sum_{k=1}^n p_1[k]\mathcal{W}(k)^{\alpha}\right)^{1/\alpha-1} - \left(\sum_{k=1}^n p_2[k]\mathcal{W}(k)^{\alpha}\right)^{1/\alpha-1}}_{1/(1-\alpha)}\cdot\norm{\mathcal{W}(j)^{\alpha}}_{1/\alpha}\nonumber\\
    &\quad=\frac{1}{\alpha}\norm{\left(\sum_{k=1}^n p_1[k]\mathcal{W}(k)^{\alpha}\right)^{1/\alpha-1} - \left(\sum_{k=1}^n p_2[k]\mathcal{W}(k)^{\alpha}\right)^{1/\alpha-1}}_{1/(1-\alpha)}\nonumber\\
    &\quad\leq \frac{1}{\alpha}\norm{\abs{\left(\sum_{k=1}^n p_1[k]\mathcal{W}(k)^{\alpha}\right) - \left(\sum_{k=1}^n p_2[k]\mathcal{W}(k)^{\alpha}\right)}^{1/\alpha-1}}_{1/(1-\alpha)}\nonumber\\
    &\quad= \frac{1}{\alpha}\norm{\left(\sum_{k=1}^n (p_1[k]-p_2[k])\mathcal{W}(k)^{\alpha}\right)}_{1/\alpha}^{1/\alpha-1}\nonumber\\
    &\quad\leq \frac{1}{\alpha}\left(\sum_{k=1}^n \abs{p_1[k]-p_2[k]}\cdot\norm{\mathcal{W}(k)^{\alpha}}_{1/\alpha}\right)^{1/\alpha-1}\nonumber\\
    &\quad= \frac{1}{\alpha}\left(\sum_{k=1}^n \abs{p_1[k]-p_2[k]}\right)^{1/\alpha-1}\nonumber\\
    &\quad= \frac{1}{\alpha}\norm{p_1-p_2}_1^{1/\alpha-1},\nonumber
  \end{align}
  where the first inequality follows from the H\"{o}lder inequality (Lemma~\ref{lemma:SchattenHolder}), the second inequality follows from Lemma~\ref{lemma:MatrixPowerDiff}, the third inequality follows from the triangle inequality, and the second and fourth equalities follow from the fact that $\mathcal{W}(j)\in\mathcal{D}_d$ for all $j$.
  This completes the proof.
\end{proof}

\section{Details of the Universal Fast Gradient Method}
\label{appendix:FGM}
Algorithm~\ref{alg:FGM} provides the pseudocode for the universal fast gradient method (FGM) \cite{Nesterov2015} as applied to the optimization problem \eqref{def:ExpRenyiCapacity}.

The convergence guarantee of FGM is originally stated as follows \cite[Theorem~3]{Nesterov2015}:
\begin{align} 
    \label{eq:FGM_original_bound}
    \tilde{g}_{\alpha}^{\text{R}}(p_{T+1}) -\tilde{g}_{\alpha}^{\text{R}}(p_{\star})
    \leq \min_{\nu,L_{\nu}}\left(\frac{2^{2+4\nu}L_{\nu}^{2}}{\varepsilon^{1-\nu}T^{1+3\nu}}\right)^{1/(1+\nu)}B_h(p_{\star},p_1)+\frac{\varepsilon}{2},
\end{align}
where $p_{\star}$ denotes the minimizer of \eqref{def:ExpRenyiCapacity}.
For the prox function $h(p)=\inner{p}{\log(p)}$, the Bregman divergence $B_h(p_{\star},p_1)$ with the initial iterate $p_1=\mathbf{1}_n/n$ can be bounded as
\begin{align}
    B_h(p_{\star},p_1)
    =\inner{p_{\star}}{\log(p_{\star}) - \log(p_1)}
    \leq \inner{p_{\star}}{-\log(p_1)}
    =\log(n).\nonumber
\end{align}
Substituting this bound into \eqref{eq:FGM_original_bound} yields the convergence guarantee presented in Lemma~\ref{lemma:FGMConv}.

Regarding the implementation, $\varepsilon>0$ remains a tunable parameter in Algorithm~\ref{alg:FGM}.
Its role in the theoretical convergence analysis and its impact on empirical performance are discussed in Sections~\ref{sec:MaximizingRenyiInfo} and \ref{sec:NumericalResults}, respectively.

\begin{algorithm}[t]
\caption{Universal Fast Gradient Method (FGM)}
\label{alg:FGM}
\begin{algorithmic}[1]
\STATE \textbf{Initialization:} $L_1=1$, $p_1 = \tilde{p}_1=\mathbf{1}_n/n$, $A_1 = 0$, $h(p)=\inner{p}{\log(p)}$, $\varepsilon>0$.
\STATE Set the function $\phi_1(p) =  \inner{\nabla \tilde{g}_{\alpha}^{\text{R}}(p_1)}{ p }$.
\FOR{$t = 1, 2, \dots, T$}
    \STATE Compute $q_t = \text{arg min}_{p \in \Delta_n} \phi_t(p)$.
    \STATE Set $i_t = 0$.
    \WHILE{true}
        \STATE Set $a_{t+1, i_t}=\frac{1+\sqrt{1 + 2^{i_t+2}A_t}}{2^{i_t+1}L_t}$, $A_{t+1, i_t} = A_t + a_{t+1, i_t}$, $\tau_{t, i_t} = \frac{a_{t+1, i_t}}{A_{t+1, i_t}}$, $\tilde{p}_{t+1, i_t} = \tau_{t, i_t} q_t + (1 - \tau_{t, i_t}) p_t$
        \STATE Compute $\hat{p}_{t+1, i_t} = T_{h,\Delta_n}(q_t, \nabla \tilde{g}_{\alpha}^{\text{R}}(\tilde{p}_{t+1, i_t}), a_{t+1, i_t})$ by \eqref{alg:EMD}
        \STATE Set $p_{t+1, i_t} = \tau_{t, i_t} \hat{p}_{t+1, i_t} + (1 - \tau_{t, i_t}) p_t$
        
        \IF{$\tilde{g}_{\alpha}^{\text{R}}(p_{t+1, i_t}) \leq \tilde{g}_{\alpha}^{\text{R}}(\tilde{p}_{t+1, i_t}) + \inner{\nabla \tilde{g}_{\alpha}^{\text{R}}(\tilde{p}_{t+1, i_t})}{p_{t+1, i_t} - \tilde{p}_{t+1, i_t}} + 2^{i_t-1} L_t \norm{ p_{t+1, i_t} - \tilde{p}_{t+1, i_t}}_1^2 + \frac{\varepsilon}{2} \tau_{t, i_t}$}
            \STATE \textbf{break}
        \ELSE
            \STATE $i_t = i_t + 1$
        \ENDIF
    \ENDWHILE
    \STATE Set $\tilde{p}_{t+1} = \tilde{p}_{t+1, i_t}$, $p_{t+1} = p_{t+1, i_t}$, $a_{t+1} = a_{t+1, i_t}$, $A_{t+1} = A_t + a_{t+1}$,  $L_{t+1} = 2^{i_t-1} L_t$
    \STATE Set the function $\phi_{t+1}(p) = \phi_t(p) + a_{t+1}  \inner{\nabla \tilde{g}_{\alpha}^{\text{R}}(\tilde{p}_{t+1})}{ p }$
\ENDFOR
\STATE \textbf{Output:} $p_{T+1}$ after $T$ iterations.
\end{algorithmic}
\end{algorithm}

}{}

\end{document}